\documentclass{article}
\usepackage{spconf,amsmath,graphicx}
\usepackage{amsmath,amssymb,amsfonts,mathtools,bm,bbm,xfrac}
\usepackage{hyperref}
\hypersetup{
    colorlinks=true,
    linkcolor=blue,
    filecolor=magenta,      
    urlcolor=cyan,
    pdfpagemode=FullScreen,
    citecolor=teal
    }
\usepackage{subfig}
\usepackage{cleveref}
\usepackage[usenames,dvipsnames]{xcolor}
\usepackage{enumitem}
\usepackage{mathbbol}
\usepackage{appendix}
\usepackage{amsthm}
\newtheorem{thm}{Theorem}
\newtheorem{lemma}[thm]{Lemma}
\newtheorem{defn}[thm]{Definition}
\crefname{lemma}{lemma}{lemmas}
\Crefname{lemma}{Lemma}{Lemmas}
\crefname{thm}{theorem}{theorems}
\Crefname{thm}{Theorem}{Theorems}
\crefname{defn}{definition}{definitions}
\Crefname{defn}{Definition}{Definitions}

\DeclareRobustCommand{\T}{\intercal}

\newcommand{\tr}[1]{\mathrm{tr}(#1)}

\newcommand{\qtilde}[0]{\tilde{\bm{q}}}

\newcommand{\XX}{\mbox{\tiny \it XX'}}

\title{Colored Noise Mechanism for Differentially Private Clustering}
\name{Nikhil Ravi$^*$, Anna Scaglione$^*$ and Sean Peisert$^\dagger$\thanks{This research was supported by the Director, Cybersecurity, Energy Security, and Emergency Response, Cybersecurity for Energy Delivery Systems program, of the U.S. Department of Energy, under contract DE-AC02-05CH11231.  Any opinions, findings, conclusions, or recommendations expressed in this material are those of the authors and do not necessarily reflect those of the sponsors of this work.}}
\address{$^*$Cornell Tech, Department of Electrical and Computer Engineering, New York, NY, USA,\\
$^\dagger$ Lawrence Berkeley National Lab, Berkeley, CA, USA.}
\usepackage[sorting=none]{biblatex}
\DeclareBibliographyAlias{url}{online}
\DeclareBibliographyAlias{webpage}{online}

\begin{document}
%
\maketitle
\begin{abstract}
The goal of this paper is to propose and analyze a differentially private randomized mechanism for the $K$-means query. The goal is to ensure that the information received about the cluster-centroids is differentially private. The method consists in adding Gaussian noise with an optimum covariance. The main result of the paper is the analytical solution for the optimum covariance as a function of the database.  Comparisons with the state of the art prove the efficacy of our approach. 
\end{abstract}
\begin{keywords}
differential privacy, clustering, colored noise mechanism
\end{keywords}
\section{Introduction}
\label{sec:intro}
Differential privacy (DP)~\cite{dwork2006calibrating} consists of randomized methods that allow to publish the output of data queries, while guaranteeing that the answers are statistically unlikely to reveal information about attributes of the data $X$.  Instead of releasing the result of the query, DP acts as a {\it guard} by perturbing randomly the query response. DP is also capable of ceasing to respond to repeated queries when a preset {\it privacy budget} is reached.  DP methods do not perturb or release the data directly. They are tailored to the specific query.

\textit{Clustering} algorithms are among the most common unsupervised learning techniques  (see e.g., \cite{xu2005survey} for a survey). In many cases, by summarizing in a small set of  data patterns the emerging trends in a large detailed collection, the clustering query can provide sufficient information to improve services offering from entities that do not have permission to observe the data directly. Just like publishing an average, publishing the average behavior of a data cluster for a certain set of data $X$ in a database ${\cal X}$ leaks private information, motivating the use of the DP framework for their release. 
With this in mind, this paper presents a novel approach for applying  to the publication of the centroids computed by the $K$-means algorithm.\\ 
{\bf Prior art} -- Generic differentially private clustering techniques have been previously presented in e.g.~\cite{balcan2017differentially,xia2020distributed,lu2020differentially}. The authors in \cite{balcan2017differentially} proposed an iterative $K$-means clustering algorithm for data in high-dimensional Euclidean spaces. In \cite{xia2020distributed}, the authors proposed a local DP iterative clustering algorithm where noise is added at the user's end before transmitting the data to the aggregator. While guaranteeing DP, these techniques may not-converge. Instead, the authors in \cite{lu2020differentially} proposed a clustering algorithm that performs an input perturbation in each iteration, which offers convergence guarantees but drives the cost of DP higher depending on the number of iteration required for convergence.\\ 
{\bf Contribution} -- The paper introduces the design of DP mechanism  to publish cluster centroids by adding to them Gaussian noise with an optimum covariance.  The proposed method is not iterative and provides greater accuracy for a given privacy budget. The efficacy of the proposed mechanism is tested on samples drawn from the Marketing Campaign dataset~\cite{marketingdataset}.\\
{\bf Paper organization}-- In \Cref{sec:prelims}, we introduce the DP framework before setting up the problem statement. In \Cref{sec:additive}, we describe a DP mechanism for the publication of the clustering query. In \Cref{sec:numericals}, we test numerically our algorithms, before concluding the paper in \Cref{sec:conclusion}.\\
{\bf Notation}-- Boldfaced lower-case (upper-case respectively) letters denote vectors (matrices respectively) and  $x_i$ ($X_{ij}$ respectively) denotes the $i$\textsuperscript{th} element of a vector $\bm{x}$ (the $ij$\textsuperscript{th} entry of a matrix $\bm{X}$ respectively). Calligraphic letters denote sets and $|\cdot|$ their cardinality. Finally, $[N]$ denotes the set of integers $\{0,1,\ldots,N-1\}$.

\section{Preliminaries and Problem Statement}\label{sec:prelims}
In the following, we denote by $X$ a set of feature vectors $\bm x_p$, $p\in [P]$ embedded in $\mathbb{R}^d$ that are in a database $\mathcal{X}$. To review the basic concepts, we set the problem in general terms, denoting by $\mathbb{q}(X)$ the function mapping $X$ onto the query answer, with outcome denoted by  $\bm{q} \in \mathcal{Q}$, where $\mathcal{Q}$ is its domain.

\subsection{Differential Privacy}\label{sec:dp_prelims}
A DP randomized algorithm applied to a given query makes it either difficult or impossible to tell if the data $X$ or $X'$, which is missing one feature vector relative to $X$, was queried. We denote the DP query answer by $\tilde{\mathbb{q}}(X)$, and has a random outcome $\qtilde \in \mathcal{Q}$, with distribution $f(\qtilde|X)$ (the probability density function for continuous random queries and the probability mass function for discrete random variables). 
We briefly introduce the conventional definitions that explain how differential privacy is measured and established. The first and the most widespread definition of differential privacy was introduced in \cite{dwork2006calibrating,dwork2006our} and is referred to as $(\epsilon,\delta)$-DP. 
The one that we follow was introduced by \cite{machanavajjhala2008privacy}, referred to as $(\epsilon,\delta)$-Probabilistic Differential privacy (PDP).  It can be shown that $(\epsilon,\delta)$-PDP is a strictly stronger condition than $(\epsilon,\delta)$-DP.
\begin{defn}[$(\epsilon,\delta)$-Probabilistic Differential privacy]\label[defn]{def:probabilisticDP} The so-called privacy leakage function  $L_{\XX}$ is the log-likelihood ratio between the two hypotheses that the query outcome $\qtilde$ is the answer generated by the data $X$ or the data $X'$ that differ by one element. Mathematically:
\begin{equation}
    L_{\XX}(\qtilde):=\log \frac{f(\qtilde|X)}{f(\qtilde|X')},
\end{equation}
  A randomized mechanism $\tilde{\mathbb{q}}(X)$ is $(\epsilon,\delta)$ PDP for $X$ iff:
  \begin{equation}\label{eq:def1}
      \sup_{X'}~Pr\left(L_{\XX}(\qtilde) >\epsilon\right)\leq \delta.
  \end{equation}
\end{defn}
\begin{thm}[PDP implies DP~\cite{mcclure2015relaxations}]\label[thm]{thm:PDP-DP}
If a randomized mechanism is $(\epsilon,\delta)$-PDP, then it is also $(\epsilon,\delta)$-DP, i.e.,
\[
    (\epsilon,\delta)-\text{PDP}  \Rightarrow (\epsilon,\delta)-\text{DP}, \text{ but } (\epsilon,\delta)-\text{DP} \nRightarrow (\epsilon,\delta)-\text{PDP}.
\]
\end{thm}


 \Cref{def:probabilisticDP} has an direct statistical interpretation: if values of $(\epsilon,\delta)$ are close to zero, even when one adopts the optimum statistical test for the hypotheses that the randomized answer is produced by the datasets $X$ or $X'$, the test produces results that mostly are incorrect or unreliable. Of course, this comes at a cost in terms of accuracy of the answer. 

\subsection{The \texorpdfstring{$K$}{K}-Means Clustering Query} \label{sec:kmeans_prelims}
Let $[P]$ be the set of indices of the $P$ data points in $X$, $\bm{x}_p \in \mathbb{R}^d, p\in [P]$, which we can organize as a  the $P\times d$ matrix $\bm{X} := [\bm{x}_1, \ldots, \bm{x}_P]^\T$. The task of the $k$-means algorithm is to split the dataset, into $K>1$ subsets (clusters) and to assign a label to each point corresponding to the nearest cluster centroid to itself. In other words, the query we are interested in is given by $\mathbb{q}~:~\mathbb{R}^{P \times d}~\rightarrow~[K]^{P}~\times~\mathbb{R}^{K \times d}$. The problem can formally be posed as an optimization problem of the form:
\begin{equation}\label{eq:clus}
  \arg\min_{\mathfrak{C}} ~~ \frac{1}{P} \sum_{\substack{k \in [K]\\
                 p \in \mathcal{C}_k}} \|\bm{x}_p - \bm{c}_k\|_p \quad \text{s.t.} \quad \bigcap_{k \in [K]} \mathcal{C}_k = [P],
\end{equation}
where $\mathfrak{C} = \{\mathcal{C}_1,\ldots,\mathcal{C}_K\}$ is the partition in $K$ clusters and $\bm{c}_k \in \mathbb{R}^d$ is the centroid of a cluster $\mathcal{C}_k$, obtained by averaging the points $\{\bm{x}_p\}_{p\in \mathcal{C}_k}$. The objective in \cref{eq:clus} minimizes the cost of clustering assignment $\mathfrak{C}$ with the constraint as shown, so that every point in the database is assigned a cluster label whose centroid is the closest. 

\section{An Optimized \texorpdfstring{$(\epsilon,\delta)$}{(e,d)}-DP Gaussian Mechanism for \texorpdfstring{$K$}{K}-Means Clustering }\label{sec:additive}
Prior to introducing our optimization in Section \ref{sec:colored-noise}, in the next section we how the most common DP mechanism for continuous queries would perform. 
\subsection{White Gaussian Noise Mechanism}\label{sec:white_gaussian}
The Gaussian noise output perturbation mechanism is a popular option in DP for publishing a variety of statistics. It entails adding a sample of i.i.d. random noise $\bm{\eta}\in \mathbb{R}^d$ prior to publishing a vector query. When applied to the cluster centroids:
\begin{equation}
    \tilde{\bm{c}} := \tilde{\mathbb{c}}(X) = \mathbb{c}(X) + \bm{\eta} = \bm{c} + \bm{\eta},\label{eq:dp_centroid}
\end{equation}
where $\bm{\eta}$ is a sample of random noise and $\tilde{\bm{c}}$ is the DP answer.

\begin{thm}[Cluster centroids are $(\epsilon,\delta)$-DP~\cite{dwork2006calibrating}]\label[thm]{thm:cluster_centroids_DP}
The mechanism in \cref{eq:dp_centroid}, when $\bm{\eta}~\sim~\mathcal{N}(\bm{0},\sigma^2\bm{I})$ provides $(\epsilon,\delta)$-DP for any two neighboring datasets $X$ and $X'$:
\begin{equation}
    \sigma = \frac{\Delta\mathbb{c}}{\epsilon} \sqrt{2\log(2/\delta)}\label{eq:color_gaussian}
\end{equation}
and $\Delta\mathbb{c}$ is the query sensitivity given by:
\begin{equation*}
\Delta\mathbb{c} = \sup_{X'}\|\mathbb{c}(X) - \mathbb{c}(X')\|_2.
\end{equation*}
\end{thm} 

\subsection{Colored Gaussian Noise Mechanism}\label{sec:colored-noise}
Our idea in this paper stems from the fact that it is possible easily to generalize the i.i.d. Gaussian noise mechanism to the case of correlated Gaussian noise:
\begin{equation}
    \tilde{\mathbb{c}}(X) = \mathbb{c}(X) + \hat{\bm{\eta}} ~~\mbox{where}~~\hat{\bm{\eta}} \sim \mathcal{N}(\bm{0}, \bm{\Gamma}^{-1})\label{eq:dp_centroid_color}
\end{equation}
where $\bm{\Gamma}$ is the noise precision matrix, and optimize $\bm \Gamma$ once the $(\epsilon,\delta)$-PDP tradeoff is computed. Intuitive, a different choice of the covariance can better capture how the centroids are collectively placed in $\mathbb{R}^d$. 
The following theorem states the privacy guarantees of this mechanism:
\begin{thm}[Colored Gaussian Noise Mechanism is $(\epsilon,\delta)$-DP]\label[thm]{thm:cluster_centroids_DP_color}
The additive noise mechanism in \cref{eq:dp_centroid_color} provides $(\epsilon,\delta)$-DP for any two neighboring datasets $X$ and $X'$.
\end{thm} 
\begin{proof}
In order to show that the mechanism in \cref{eq:dp_centroid_color} is $(\epsilon,\delta)$-DP, we first show that it is $(\epsilon,\delta)$-PDP. Consider:
\begin{align}
    L_{\XX}(\tilde{\bm{c}}) =  \log \frac{f(\tilde{\mathbb{c}}(X)|X)}{f(\tilde{\mathbb{c}}(X')|X')} = \log \frac{f(\tilde{\bm{c}} - \mathbb{c}(X)|X)}{f(\tilde{\bm{c}} - \mathbb{c}(X')|X')}.
    \end{align}
    With $\tilde{\bm{c}} = \mathbb{c}(X) + \hat{\bm{\eta}}$ and $\hat{\bm{\eta}} \sim \mathcal{N}(\bm{0}, \bm{\Gamma}^{-1})$, we have:
    \begin{align}
        &L_{\XX}(\tilde{\bm{c}}) = \log \frac{\exp^{\left[-\frac{1}{2} \left(\tilde{\bm{c}} - \mathbb{c}(X)\right)^\T \bm{\Gamma} \left(\tilde{\bm{c}} - \mathbb{c}(X)\right)\right]}}{\exp^{\left[-\frac{1}{2} \left(\tilde{\bm{c}} - \mathbb{c}(X')\right)^\T \bm{\Gamma} \left(\tilde{\bm{c}} - \mathbb{c}(X')\right)\right]}}\nonumber\\
        &= \left(\mathbb{c}(X)-\mathbb{c}(X')\right)^\T\bm{\Gamma} \hat{\bm{\eta}} + \frac{\|\bm{\Gamma}^{\frac{1}{2}}\left(\mathbb{c}(X)-\mathbb{c}(X')\right)\|_2^2}{2}.
    \end{align}
    The privacy loss function is a linear transformation of a Gaussian random vector and thus, it is a Gaussian random variable with expectation $0.5 \cdot \|\bm{\Gamma}^{\frac{1}{2}}\left(\mathbb{c}(X)-\mathbb{c}(X')\right)\|_2^2$ and variance $\|\bm{\Gamma}^{\frac{1}{2}}\left(\mathbb{c}(X)-\mathbb{c}(X')\right)\|_2^2$.
In order to prove $(\epsilon,\delta)$-PDP, we have to prove that the privacy leakage function exceeds $\epsilon$ with probability at most $\delta$, i.e.:
\begin{align}
     1 \geq \frac{\Delta\mathbb{c}}{\epsilon} \sqrt{2\log(2/\delta)},\label{eq:color_gaussian_a}
\end{align}
where $\Delta_c$ is the local sensitivity given by:
\begin{align}
    \Delta_c :=  \sup_{X'} \|\bm{\Gamma}^{\frac{1}{2}}\left(\mathbb{c}(X)-\mathbb{c}(X')\right)\|_2.
\end{align}
Thus, the colored noise additive mechanism in \cref{eq:dp_centroid_color} is $(\epsilon,\delta)$-PDP with mean $\bm{0}$ and covariance $\bm{\Gamma}^{-1}$. Finally, the said mechanism is also $(\epsilon,\delta)$-DP from \Cref{thm:PDP-DP}.
\end{proof}

The design of the optimal noise vector hinges on the design of its covariance matrix. Let us define:
\begin{align*}
    \gamma_c\!&:=\!\frac{\epsilon^2}{2\log\left(\frac{2}{\delta}\right)}\!\geq\!\Delta\mathbb{c}^2 \!\geq\!\left(\mathbb{c}(X)\!-\!\mathbb{c}(X')\right)^{\T}\!\bm{\Gamma}\left(\mathbb{c}(X)\!-\!\mathbb{c}(X')\right).
\end{align*}
Minimizing the mechanism error (or the distortion) , meeting the DP guarantees, requires trace ($\mathrm{tr}$) of the noise covariance (the inverse of $\bm \Gamma$), i.e. solving the following optimization:
\begin{align}
    \min_{\bm{\Gamma}}~&\tr{\bm{\Gamma}^{-1}}\label{eq:color_optimization}\\
    \text{s.t.}~&\left(\mathbb{c}(X)-\mathbb{c}(X')\right)^{\T} \bm{\Gamma} \left(\mathbb{c}(X)-\mathbb{c}(X')\right) \leq \gamma_c, \hfill \forall X'\in {\cal X}.\nonumber
\end{align}
In the following lemma, we provide the closed form solution:
\begin{lemma}[Optimal Choice of $\bm{\Gamma}$]\label[lemma]{lem:optimal_Phi}
Let the matrix $\bm C_{\XX}$ contain as its columns all possible $(\mathbb{c}(X)-\mathbb{c}(X'))$, $X'\in {\cal X}$ and let us assume that $\bm C_{\XX}$ is full row rank. Let us assume that the first $Kd$ columns of $\bm C_{\XX}$, corresponding to the set ${\cal D}\subseteq \cal X$ have the smallest norms and are linearly independent, forming the matrix we refer to as $\bm C_{\XX}^*$. The optimization problem in \cref{eq:color_optimization} has a unique solution  and it evaluates to:
\begin{equation}
    \bm{\Gamma}^\star = {\bm{R}}_{\bm{\lambda}^\star}^{-\frac{1}{2}}
\end{equation}
where $\bm \lambda^*$ has only $Kd$ non-zero values which correspond to the constraints associated with the set ${\cal D}$ and:
\begin{equation}
    \bm{R}_{\bm{\lambda}^\star}\!:=\!\sum_{X' \in \mathcal{D}} {\lambda}_{X'}^\star \left(\mathbb{c}(X)\!-\!\mathbb{c}(X')\right)\left(\mathbb{c}(X)\!-\!\mathbb{c}(X')\right)^{\T},
\end{equation}
where $\lambda_{X'}^\star, \forall X'\in {\cal D}$ are the non-zero Lagrange multipliers for the problem in \cref{eq:color_optimization}. Their values are:
\begin{align}
    \lambda^*_i =v_i^{-2} ~~\mbox{where}~~\bm v\triangleq \gamma_c\bm M^{-1}\bm 1, ~~ M_{ij}\triangleq[\bm C_{\XX}^{*\frac{\T}{2}}]^2_{ij}.
\end{align}

\end{lemma}
\begin{proof}
The Lagrangian of the optimization in \cref{eq:color_optimization} is:
\begin{align}
    \mathfrak{L}
    &= \tr{\bm{\Gamma}^{-1}} \nonumber\\
    & ~~~+ \sum_{X' \in \mathcal{X}} \!\lambda_{X'} \left[ \tr{\bm{\Gamma}\left(\mathbb{c}(X)\!-\!\mathbb{c}(X')\right) \left(\mathbb{c}(X)\!-\!\mathbb{c}(X')\right)^{\T}}\!-\! \gamma_c\right]\nonumber\\
    &= \tr{\bm{\Gamma}^{-1} + \bm{\Gamma} \bm{R}_{\bm{\lambda}} - \gamma_c\mathrm{diag}(\bm{\lambda})}, \qquad\text{where}
\end{align}
\begin{align}
    \bm{R}_{\bm{\lambda}} := \sum_{X' \in \mathcal{X}} \lambda_{X'} \left(\mathbb{c}(X)\!-\!\mathbb{c}(X')\right)\left(\mathbb{c}(X)\!-\!\mathbb{c}(X')\right)^{\T}.
\end{align}
Let $\bm C_{X X'}$ be the $Kd\times P$ matrix containing all vectors $\left(\mathbb{c}(X)-\mathbb{c}(X')\right)$ for all $X\in {\cal X}$. Note that 
\begin{equation}
    \bm{R}_{\bm{\lambda}}= \bm C_{X X'}\mathrm{diag}(\bm \lambda) \bm C_{X X'}^\T,
\end{equation}
where $\bm \lambda$ contains the Lagrange multipliers. 
The problem has a unique solution if $\bm{R}_{\bm{\lambda}}$ is invertible at the optimum $\bm \lambda$. A necessary condition is that $|\mathcal{X}|=P\geq Kd$ and that there are at least $Kd$ non-zero Lagrange multipliers at the optimum point. 
In this case, the stationary point of the problem is:
\begin{align}
    \frac{\nabla \mathfrak{L}(\bm{\Gamma},\bm{\lambda})}{\nabla \bm{\Gamma}} = 0 \quad \Rightarrow \quad \bm{\Gamma}^\star =  \bm{R}_{\bm{\lambda}^\star}^{-\frac{1}{2}}.\label{eq:color_precision}
\end{align}
Substituting we have:
\begin{align}
    \mathfrak{L}\left(\bm{R}_{\bm{\lambda}^\star}^{-\frac{1}{2}},\bm{\lambda}^\star\right) = \tr{2\bm{R}_{\bm{\lambda}^\star}^{\frac{1}{2}} - \gamma_c\mathrm{diag}(\bm{\lambda}^\star)}.
\end{align}
which makes it clear that the non-zero multipliers (those for which the constraint is tight) should be the smallest columns in $\bm C_{\XX}$ in terms of $L_2$-norm. Since we need at least $Kd$ of them, in $\mathcal{D}\subseteq {\cal X}$ that are linearly independent, we can place them in the matrix ${\bm C}^*_{\XX}$, which is square and invertible, and assume  that these are the first $Kd$ columns of $\bm C_{\XX}$ without loss of generality. 
The individual constraints in \cref{eq:color_optimization}  $\forall X' \in \mathcal{D}$:
\begin{align}
    \left(\mathbb{c}(X)\!-\!\mathbb{c}(X')\right)^{\T}\bm{R}_{\bm{\lambda}}^{-\frac{1}{2}}\left(\mathbb{c}(X)\!-\!\mathbb{c}(X')\right) = \gamma_c\label{eq:color_exact_constraint}
\end{align}
can be rewritten as:
\begin{align}
    &\gamma_c \bm{1} = \mathrm{diag}\left( {\bm{C}_{\XX}^{*\T}} \left({\bm{C}_{\XX}^{*}} \mathrm{diag}(\bm{\lambda}_+^{*}){\bm{C}_{\XX}^{*\T}}  \right)^{-\frac{1}{2}} \bm{C}_{\XX}^{*}\right)\nonumber\\
    &=  \mathrm{diag}\left( {\bm{C}_{\XX}^{*\frac{\T}{2}}} \left(\mathrm{diag}(\bm{\lambda}^{*}_+)\right)^{-\frac{1}{2}} {\bm{C}_{\XX}^{*\frac{1}{2}}}\right)
\label{eq:color_lambda_solution}
\end{align}
and the solution of \eqref{eq:color_lambda_solution} will give $\bm \lambda^*=(\bm \lambda^{*\T}_+,\bm 0^\T)$.
With some algebra, one can express \eqref{eq:color_lambda_solution} as follows:
\begin{align}
   \gamma_c\bm 1&=\bm M \bm v,
\end{align}
where the entries of $\bm v$ are $v_{i}=\lambda_i^{-\frac 1 2},~i\in [Kd]$ and those of $\bm M$ are
$M_{ij}=[{{\bm C}_{\XX}^{*\frac{\T} 2 }}]^2_{ij}$.
Solving for $\bm v=\gamma_c\bm M^{-1}\bm 1$, one obtains the entries of $\bm \lambda_{+}^*$ as $v_i^{-2},~i\in [Kd]$ and can calculate the optimal $\bm \Gamma$ that satisfies the conditions from \cref{eq:color_precision}.
\end{proof}

\section{NUMERICAL SIMULATIONS}\label{sec:numericals}
To illustrate the methodology proposed, we consider the Marketing Campaign dataset~\cite{marketingdataset} which contains data about when $P = 2212$ customers accepted offers of five marketing campaigns with additional information such as income, size of household, amount spent on various products, among others (see~\cite{marketingdataset} for a detailed description). The dataset contains a total of $d=28$ features (excluding the ID of a customer) with 3 categorical columns (marital status, education, and date of registration) which were encoding into numeric form. The objective of the exercise is to predict the demographics of the customers who will respond to a marketing campaign and thereby increasing profits. Such campaigns have become commonplace, and often violate customer privacy~\cite{tucker2014social}. We first find that there are four clusters in this dataset using the elbow method, and this is further reiterated in the scatter plot of the points (embedded in 2-d using Multidimensional Scaling) in \cref{fig:mds}. In \cref{tab:count}, we show the population count of the individual clusters. We then add colored noise to the cluster centroids and reevaluate the labels by assigning the closest noisy cluster centroid to a point as its new label. In \cref{fig:spe_inc_noisy}, we show the noisy clusters in the plot of Customer Spending against their income. The classes can be characterized as medium-income and low-spending (class 0), medium-income and medium-spending (class 1), high-income and high-spending (class 2), and low-income and low-spending (class 3). 
\begin{table}[]
    \centering
    \begin{tabular}{|c|c|c|c|c|}
        \hline
         \textbf{Cluster} &  \textbf{0} & \textbf{1} & \textbf{2} & \textbf{3}\\
         \hline
         \textbf{True} & 586 & 510 & 506 & 610 \\
         \hline
         \textbf{Noisy} & 581 & 483 & 459 & 679 \\
         \hline
    \end{tabular}
    \caption{True and Noisy clustering population count.}
    \label{tab:count}
\end{table}
In \cref{fig:dp}, we compare the performance of white and colored noise mechanisms for various $(\epsilon,\delta)$ pairs, and we observe a performance improvement for the colored noise mechanism over the white noise mechanism and the mechanism mentioned in \cite{balcan2017differentially}\footnote{For the mechanism from \cite{balcan2017differentially}, the $\delta=1$ curve corresponds to $\delta = 0.999$.}. 
Finally, in \crefrange{fig:deals}{fig:promos}, we plot the distribution of the number of deals purchased and total number of promotions accepted by customers belonging to the different clusters, with and without noise added. We can observe that the DP colored noise does not affect the distribution of these counts. As such, the noisy clustering mechanism will still lead to highly accurate inferences that may be drawn from these parameters, such as the conclusion that married customers accept deals and promotions in relatively higher numbers, that the high-income and high-spending customers do not care for deals but are more likely to accept promotion, and that customers in class 0 are more likely to accept deals.
\begin{figure}
     \centering
     \subfloat[][]{\includegraphics[width=0.233\textwidth]{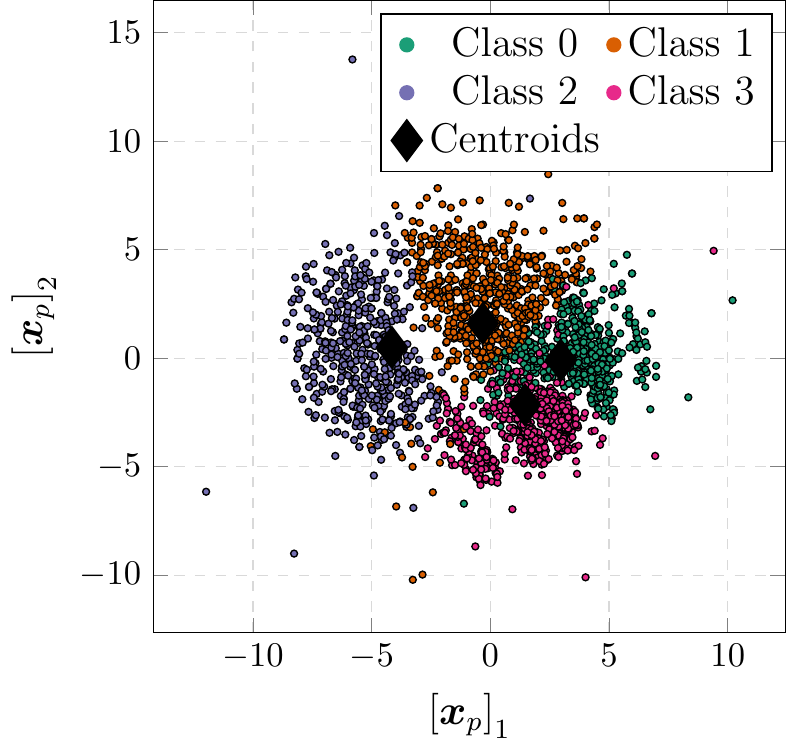}\label{fig:mds}}\quad
     \subfloat[][]{\includegraphics[width=0.227\textwidth]{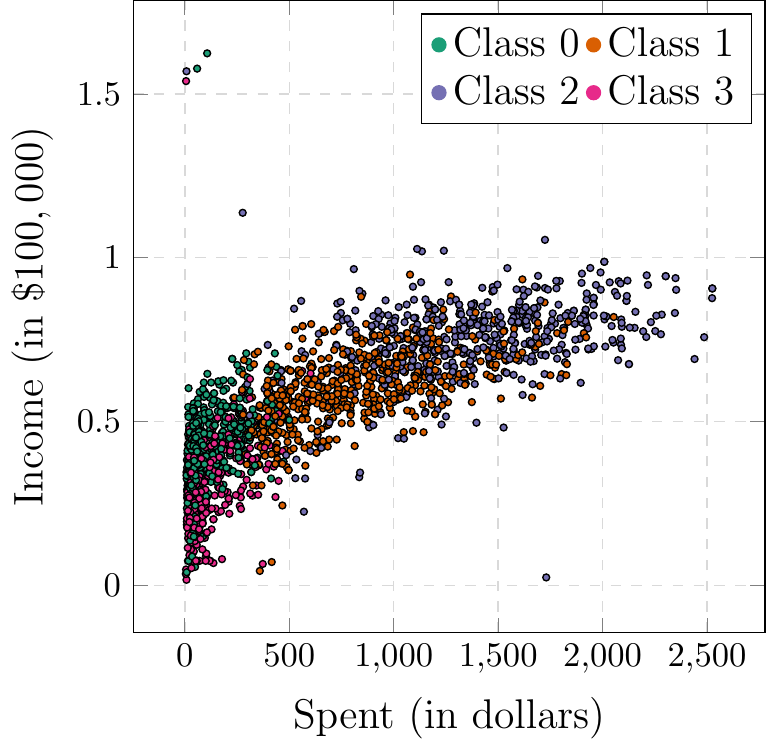}\label{fig:spe_inc_noisy}}
     \caption{(a) Multidimensional Scaling visualization of the data in 2d. (b) Scatter plot of Customer Spending vs Income.}
     \label{fig:scatter}
\end{figure}
\begin{figure}
     \centering
     \includegraphics[width=0.45\textwidth]{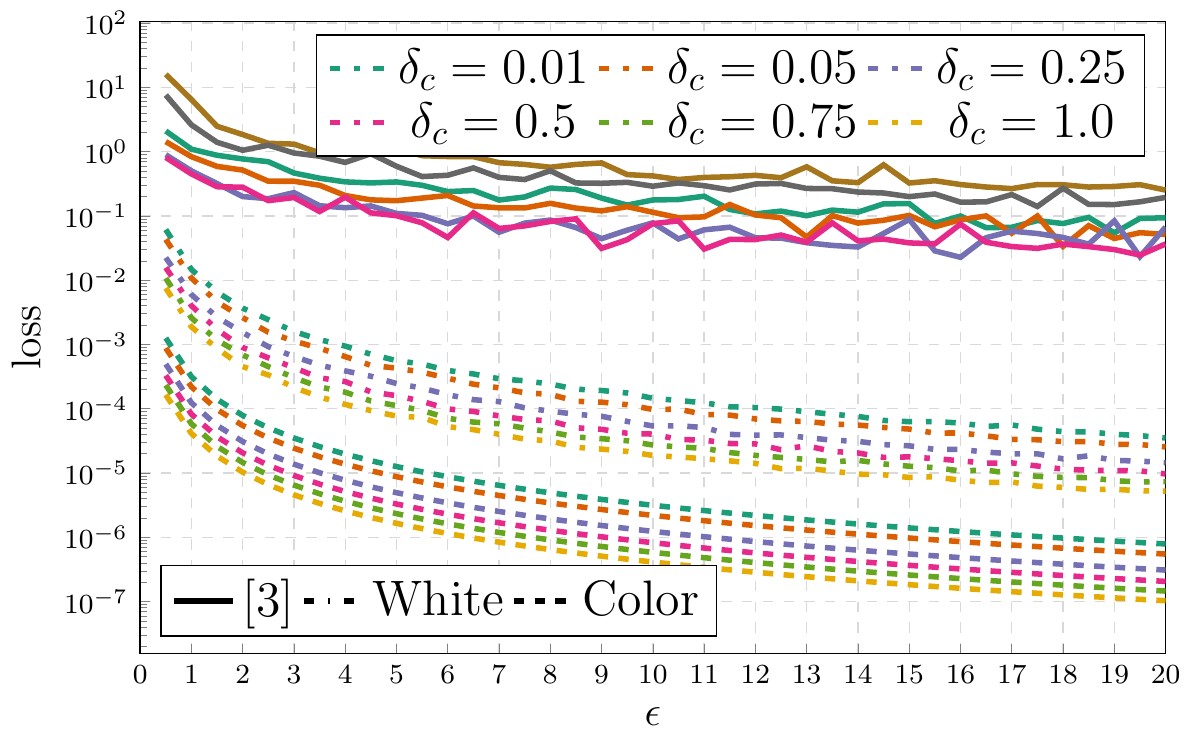}
     \caption{Fractional clustering loss vs $\epsilon$ for various schemes.}
     \label{fig:dp}
\end{figure}
\begin{figure}
     \centering
     \subfloat[][]{\includegraphics[width=0.31\textwidth]{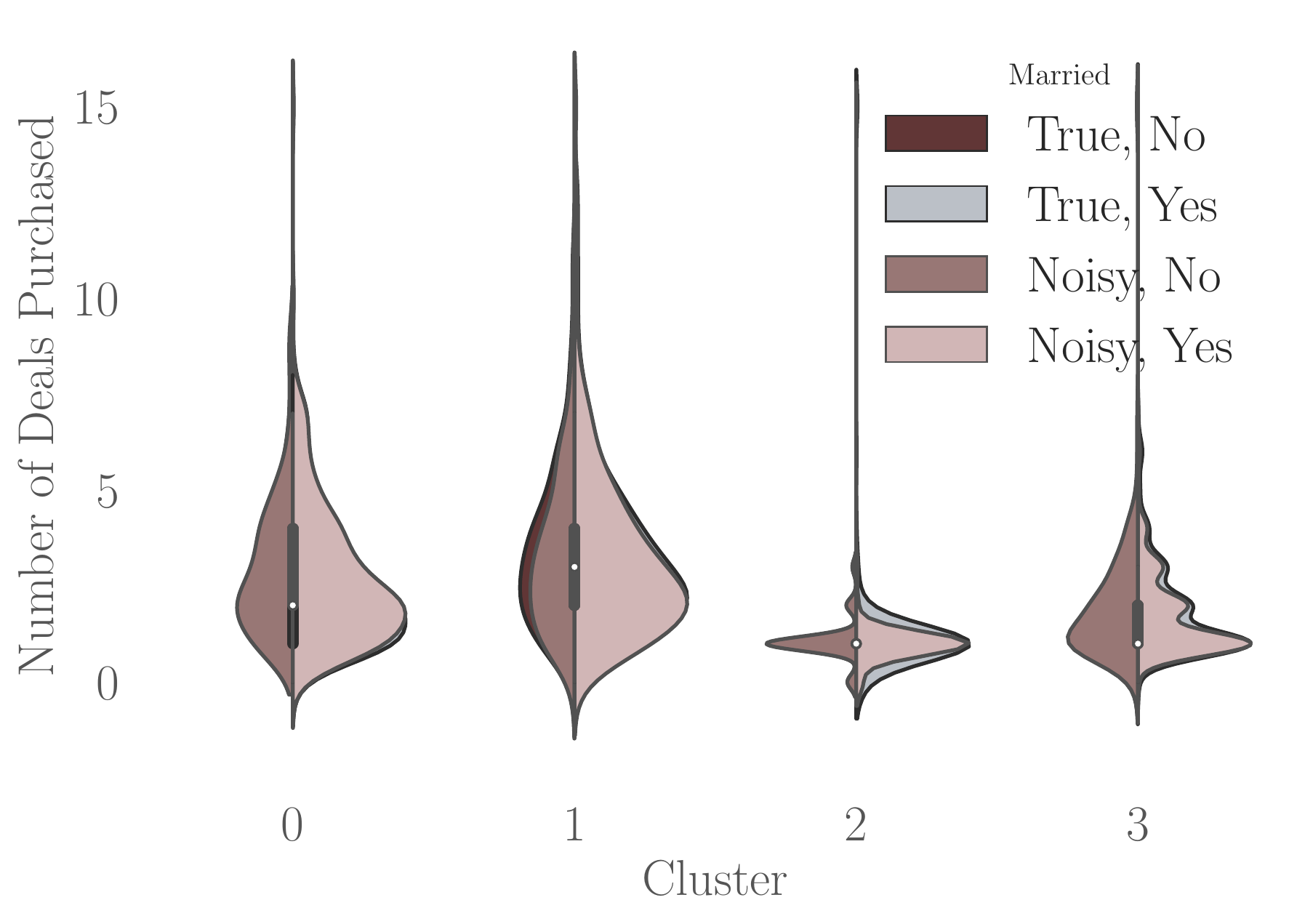}\label{fig:deals}}\\
     \subfloat[][]{\includegraphics[width=0.31\textwidth]{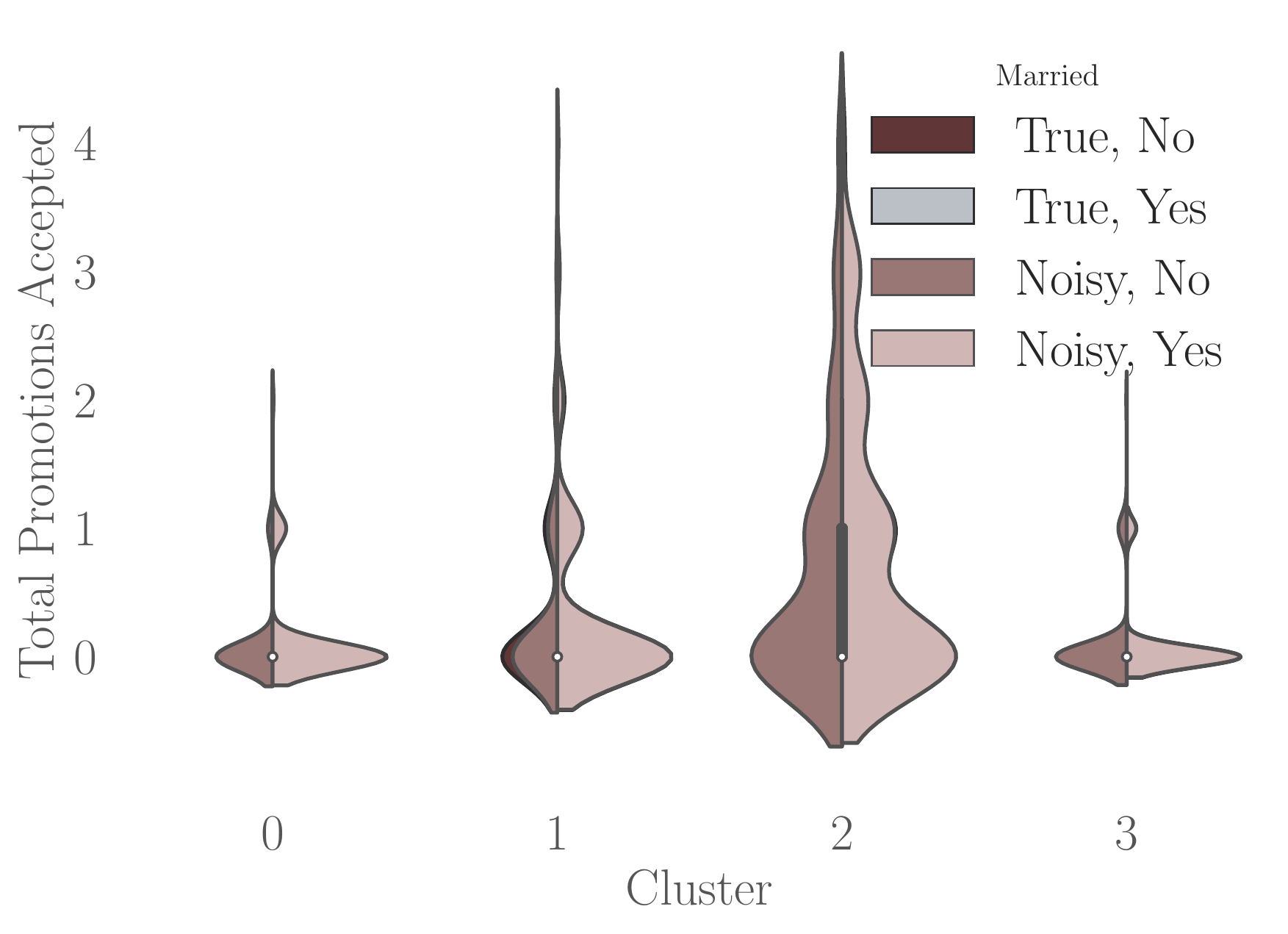}\label{fig:promos}}
     \caption{True and Noisy violin plots of (a) `Number of Deals Purchased' and (b) `Promotions accepted' across clusters.}
     \vspace{-1em}
     \label{fig:inf}
\end{figure}

\section{CONCLUSION}\label{sec:conclusion}
In this paper, we proposed and analyzed a differentially private randomized mechanism for the $K$-means clustering query. The method consisted of adding Gaussian noise with an optimum covariance. The method outperforms the traditional Gaussian noise mechanism, and existing iterative DP clustering methods, as shown via numerical simulations against a marketing campaign dataset. Finally, we show that the mechanism preserves the count distributions of the various metrics of the dataset, thereby leading to accurate inferences (relative to inferences drawn from non-noisy clustering). 

\vfill\pagebreak
\printbibliography

\end{document}